\documentclass[10pt,conference]{IEEEtran}

\usepackage{braket}
\usepackage{amsmath,amssymb,amsfonts,dsfont,amsthm}
\usepackage{algorithmic}
\usepackage{graphicx}
\usepackage{textcomp}
\usepackage{xcolor}
\usepackage{comment}
\usepackage{cite}

\newcommand\blfootnote[1]{%
  \begingroup
  \renewcommand\thefootnote{}\footnote{#1}%
  \addtocounter{footnote}{-1}%
  \endgroup
}

\usepackage{flushend}

\newtheorem{theorem}{Theorem}
\newtheorem{corollary}[theorem]{Corollary}

\newtheorem{lemma}[theorem]{Lemma}
\newtheorem{definition}[theorem]{Definition}
\newtheorem{proposition}[theorem]{Proposition}

\newlength{\blank}
\newcommand{\RR}{\mathbb{R}}
\newcommand{\NN}{\mathbb{N}}

\newcommand{\cB}{\mathcal{B}}

\newcommand{\cD}{\mathcal{D}}

\newcommand{\cG}{\mathcal{G}}

\newcommand{\cN}{\mathcal{N}}

\newcommand{\cT}{\mathcal{T}}

\newcommand{\cX}{\mathcal{X}}
\newcommand{\cY}{\mathcal{Y}}

\newcommand{\pau}[1]{{\color{purple}#1}}

\begin{document}

%\title{Galaxy codes: To infinity and beyond}
\title{Optimal Codes for Deterministic Identification over Gaussian Channels: Closing the Capacity Gap}
%IF MY BERNOULLI IDEA WORKS (GALAXY INSIDE A CUBE) THEN THE TITLE COULD BE JUST: 
%\title{Optimal Galaxy Codes: Closing the Rate Gap in Deterministic Identification}

\author{\IEEEauthorblockN{Pau Colomer}
\IEEEauthorblockA{\textit{Chair of Theoretical}\\
\textit{Information Technology}\\
\textit{Technische Universit\"at}\\
\textit{M\"unchen, Germany}\\
pau.colomer@tum.de}
\and
\IEEEauthorblockN{Christian Deppe\textsuperscript{1}}
\IEEEauthorblockA{\textit{Institute for Communications}\\
\textit{Technology}\\
\textit{Technische Universit\"at}\\
\textit{Braunschweig, Germany}\\
christian.deppe@tu-bs.de}
\and
\IEEEauthorblockN{Holger Boche\textsuperscript{1,2}}
\IEEEauthorblockA{\textit{Chair of Theoretical}\\
\textit{Information Technology}\\
\textit{Technische Universit\"at}\\
\textit{M\"unchen, Germany}\\
%\textit{6G-life \& MCQST \& MQV}\\
boche@tum.de}
\and
\IEEEauthorblockN{Andreas Winter\textsuperscript{3}}
\IEEEauthorblockA{\textit{Department Mathematik/Informatik}\\
\textit{Abteilung Informatik}\\
\textit{Universit\"at zu K\"oln}\\ 
\textit{Germany}\\
andreas.winter@uni-koeln.de}
}

\date{18 February 2026} 

\maketitle

\pagestyle{plain}

\begin{abstract}
Deterministic identification (DI) has emerged as a promising paradigm for large-scale and goal-oriented communication systems. Despite significant progress, a fundamental open problem has remained unresolved: a persistent gap between the best known lower and upper bounds on the DI capacity, as well as on the corresponding rate–reliability tradeoff bounds. 
In this paper, we finally close this gap for Gaussian channels $\cG$ by constructing an optimised code that achieves the known upper bound. This allows us to establish that the linearithmic capacity for deterministic identification is $\dot C_{\text{DI}}(\cG)=\frac12$. Furthermore, we analyse the rate–reliability tradeoff and show that the proposed scheme matches the known upper bounds to first order, thereby closing the existing gap in reliability performance for all admissible error decay regimes. Finally, we demonstrate the existence of an optimum universal code, which does not require knowledge of the channel parameters and yet achieves capacity.
\end{abstract}


\blfootnote{
{\textsuperscript{1}6G-life, 6G research hub,
Braunschweig and M\"unchen,
Germany.}\vspace{1pt}

{\textsuperscript{2}Cluster of Excellence, CeTI, Technische Universität Dresden, Germany.}\vspace{1pt}

{\textsuperscript{3}ICREA {\&} Universitat Aut\`onoma de Barcelona, Barcelona, Spain.}\vspace{1pt}
}

\section{Introduction}\label{sec:intro}
\subsection{Motivation and Background}\label{ssec:motivation}
As communication systems grow in size and complexity, supporting massive numbers of users, devices, and interactions becomes a central challenge. In such large-scale settings, efficient protocols must go beyond the classical Shannon paradigm of reliable message transmission \cite{Shannon:TheoryCommunication} and instead focus on accomplishing specific tasks with minimal resource expenditure. Within this perspective, goal-oriented communication paradigms, where the objective is not to reconstruct all transmitted data but to enable particular decisions or actions at the receiver, arise as a promising foundation for next-generation networks \cite{Goal_Oriented_Coms,6G_Book,6G_Book_2}.
In particular, the \emph{identification} problem, introduced in \cite{AD:ID_ViaChannels}, provides a framework that departs fundamentally from Shannon’s transmission paradigm and offers the potential for dramatic gains in communication efficiency. Instead of reconstructing the transmitted message, the receiver selects a particular message $m$ and is only required to decide reliably whether the sent message equals $m$ or not. Thus, for each possible message, only a dedicated binary hypothesis test needs to be performed, and its outcome either confirms or rejects the queried identity. 

While the size of the $M$ messages that can be transmitted grows linearly with the number $n$ of channel uses $\log M\sim nR$, it can grow exponentially for the $N$ messages that can be identified $\log N\sim 2^{nR}$ \cite{AD:ID_ViaChannels, HanVerdu:ID, RI_gaussian, Wiretap_ID, broadcast_ID, hashing_ID, CR_ID}. In other words, with the same number of channel uses, one can identify exponentially more messages than can be transmitted under Shannon’s paradigm. This striking gain, however, relies crucially on the use of stochastic encoders, which are undesirable or infeasible in many settings. It is therefore natural to investigate \emph{deterministic identification (DI)}, where randomization is not allowed.

It was observed early on that DI over discrete memoryless channels (DMCs) could only achieve linearly scaling codes, $\log N_{DMC}\sim nR$ \cite{AD:ID_ViaChannels,AC:DI}, albeit with rates higher than those achievable for transmission \cite{SPBD:DI_power}; in contrast, for some continuous-alphabet channels, the scaling can be superlinear, in fact \emph{linearithmic}: $\log N\sim R\,n\log n$ \cite{SPBD:DI_power,DI-poisson_mc,DI-fading}. This behaviour was later shown to be a universal feature of DI codes over general memoryless channels with continuous inputs and discrete outputs \cite{CDBW:DI_classical}, thereby establishing a broad class of channels for which DI asymptotically supports larger message sets than Shannon-style transmission schemes. 

A central theoretical open problem in this field is the persistent gap observed between the lower and upper bounds on the DI linearithmic rates 
\(\dot{R} := \frac{\log N}{n \log n}\) and the corresponding capacities 
\(\dot{C}_{\text{DI}}\) (the supremum of all achievable rates) over all channels with continuous input. In particular, in the general study of DI over channels \(W\) with discrete outputs \cite{CDBW:DI_classical}, the following gap was observed in the capacity bounds:
\begin{equation}\label{eq:C_general}
\frac{1}{4} d_M \leq \dot{C}_{\text{DI}}(W) \leq \frac{1}{2} d_M,
\end{equation}
where \(d_M\) is a dimension parameter (the Minkowski dimension) characterizing the geometry of the channel output.

Gaussian channels \(\mathcal{G}\), which have continuous outputs and therefore fall outside the discrete-output framework of \cite{CDBW:DI_classical}, were shown to exhibit a similar gap
\(
\frac{1}{4} \leq \dot{C}_{\text{DI}}(\mathcal{G}) \leq \frac{1}{2}
\)
in \cite{SPBD:DI_power,DI-fading}. Recently, a new code construction \cite{galaxy-codes} tightened the lower bound to 
\begin{equation}\label{eq:C_galaxy_old}
\frac38\leq\dot{C}_\text{DI}(\cG)\leq\frac12.
\end{equation}
In this paper, we close this gap by constructing an optimised code, inspired by \cite{galaxy-codes}, that achieves the known upper bound.
Furthermore, in Section~\ref{ssec:universal}, we will show that the capacity can also be achieved with a universal code, i.e., a code that does not depend on channel parameters. That means it is asymptotically optimal for any possible Gaussian channel.

%%%%%%%%%%%%%%%%%%%%%

While capacity results characterize the largest asymptotically achievable rates, they do not capture the finer structure of error behaviour. In particular, they provide no information about how rapidly the error probabilities decay as a function of the block length. This aspect is of fundamental importance in practical communication systems, where one typically encounters a tradeoff between the transmission rate and the reliability, quantified through the decay exponent of the error probabilities.

For DI, the rate–reliability tradeoff was first analysed for channels with continuous input and discrete output in \cite{CDBW:Reliability-TCOM, DI-steins}, and then extended to general linear Gaussian channels in \cite{RRGauss-arXiv}. A common feature of these works is the presence of a non-vanishing gap, analogous to that in Equation~\eqref{eq:C_general}, between the lower and upper bounds on the rate–reliability function.

In Section~\ref{sec:RR_tradeoff}, we study the rate–reliability tradeoff of the coding scheme proposed in Section~\ref{sec:main_results}. We show that the upper bound derived in \cite{RRGauss-arXiv} is achieved to first order in the regime of small error exponents. This establishes that the proposed construction is not only capacity-achieving in the asymptotic sense, but also first-order optimal with respect to the reliability function across the admissible error scaling regimes that permit a linearithmic behaviour.

\subsection{Problem setting and preliminaries}\label{ssec:setting}
We consider memoryless point-to-point communication over $n$ uses of the 
\emph{additive white Gaussian noise (AWGN) channel}. The AWGN channel occupies a central role in information theory and communications \cite{AWGN_IT, AWGN_radio, AWGN_wireless}. It serves as the canonical model for thermal noise in electronic systems and provides an accurate abstraction for a wide range of practical communication scenarios. Moreover, many more complex channel models reduce to Gaussian noise models in high-dimensional regimes or after appropriate processing. For these reasons, understanding the fundamental limits of DI over the AWGN channel is not only of theoretical interest, but also essential for potential applications of identification-based communication paradigms in realistic systems.

Formally, the channel acts on an input sequence $x^n=(x_1,\dots,x_n)\in\mathbb{R}^n$ and produces an output sequence $y^n\in\mathbb{R}^n$, which is a realization of the following sequence of random variables
\begin{equation}\label{eq:AWGN}
Y^n = x^n + \sigma Z^n,
\end{equation}
where $Z^n=(Z_1,\dots,Z_n)$ has independent and identically distributed (i.i.d.) components $Z_i \sim \mathcal{N}(0,1)$ and $\sigma>0$ denotes the noise standard deviation. Equivalently, the conditional output distribution is
\begin{equation}\label{eq:output_distr}
 \cG_{\vec{x}}:=\mathcal{G}(\cdot \mid x^n) = \mathcal{N}(x^n,\sigma^2 I_n),
\end{equation}
where $I_n$ is the $n\times n$ identity matrix.

In practical communication systems, arbitrarily large signal amplitudes 
cannot be transmitted. This physical limitation is captured by a block 
power constraint $P>0$, requiring that every admissible input sequence 
satisfy
\begin{equation}\label{eq:power_constraint}
\|x^n\|^2 \le nP,
\end{equation}
where $\|x^n\| = \left(\sum_{i=1}^n x_i^2\right)^{1/2}$ is the Euclidean norm. 
Thus, the set of admissible inputs is the closed $n$-dimensional ball of radius $\sqrt{nP}$ centred at the origin $o=(0,\dots,0)$:
\begin{equation}
\mathcal{B}_{o}(\sqrt{nP},n) 
:= \left\{ x^n \in \mathbb{R}^n : \|x^n\| \leq \sqrt{nP} \right\}. 
\end{equation}
For notational convenience, and in view of the geometric nature of our 
arguments, we define $\vec{x}\in\mathbb{R}^n$ as the vector that goes from the origin $o$ to the point $x^n$. Moreover, we will take advantage of this duality and sometimes use (or rather, abuse) the vector notation even when referring to words and sequences.

The Gaussian cumulative distribution function is denoted for all $x\in\RR$ by:
\begin{equation}\label{eq:comulative}
    \Phi(x)=\int_{-\infty}^x\frac{1}{\sqrt{2\pi}}e^{-z^2/2}dz.
\end{equation}

\begin{definition}\label{def:DI_code}
An $(n,N,\lambda_1,\lambda_2)$ DI code over $n$ uses of an AWGN channel $\cG$ is a family of pairs $\{(u_i,\cD_i) : i\in[N]\}$ with $u_i\in\cX^n$ input sequences and $\cD_i\subset\cY^n$ output regions such that for all distinct $i, j\in [N]$, and $\lambda_1,\lambda_2>0$:
\[
  \cG_{u_i}(\cD_i)\ge 1-\lambda_1 \quad\text{and}\quad \cG_{u_i}(\cD_j)\le\lambda_2.
\]
\end{definition}

The parameters $\lambda_1$ and $\lambda_2$ bound respectively the two errors that can be committed in an identification scheme: the first type is a \emph{missed identification}, when the receiver wants to identify the message sent but the outcome of the test is negative; and the second is a \emph{false identification}, when the messages are different but the test is positive. To avoid trivialities, it is also assumed that $\lambda_1+\lambda_2<1$.

We say that an $(n,N,\lambda_1,\lambda_2)$ DI code is \emph{good} if the errors vanish for sufficiently large $n$, i.e. $\lim_{n\to\infty}\lambda_1,\lambda_2 = 0$.
A number $\dot{R}> 0$ is said to be an \emph{achievable linearithmic rate} if for every $\epsilon>0$
there exists an $n_0\in\mathbb{N}$ and a sequence of good
$(n,N^{(n)},\lambda_1^{(n)},\lambda_2^{(n)})$ DI codes defined for all $n\ge n_0$, 
such that
\begin{equation}
\lim_{n\to\infty}\lambda_1^{(n)} = 0,
\,
\lim_{n\to\infty}\lambda_2^{(n)}=0\,\,\, \text{and}\,\,\,\frac{\log N^{(n)}}{n\log n} \ge \dot{R}-\epsilon.
\end{equation}

The linearithmic capacity of the channel is thus defined as
\begin{equation}\label{eq:def_capacity}
\dot{C}_{\mathrm{DI}}(\cG)=
\sup\bigl\{\dot{R}\,:\,\dot{R}\text{ is achievable}\bigr\}.
\end{equation}

Finally, let $\vec{a},\vec{v}\in\RR^n$ be two vectors and let $\|\vec{a}\|>0$, then the orthogonal projection of $\vec{v}$ onto $\vec{a}$ is denoted as
\begin{equation}\label{eq:proj_definition}
  \Pi_{\vec{a}}\vec{v} := \frac{(\vec{a},\vec{v})}{\|\vec{a}\|^2} \vec{a},
\end{equation}
with $(\vec{a},\vec{v})=\sum_{i=1}^na_iv_i$ the inner product of the two vectors.

\subsection{Intuition for the new construction}\label{ssec:intuition}
As the construction of our code is technically involved, we include here a discussion of the underlying strategy and how it differs from previous approaches. Our goal is to provide the geometric and conceptual intuition that guides the design, so that the subsequent technical lemmas and parameter choices appear as natural consequences of an optimal geometric configuration, rather than as ad hoc artifacts of the analysis.

\medskip

\noindent\textbf{Typicality-based constructions.} 
The classical DI code constructions for Gaussian channels \cite{SPBD:DI_power,DI-fading,RRGauss-arXiv} — and, more generally, the abstract constructions in \cite{CDBW:DI_classical,CDBW:Reliability-TCOM} — are based on typicality arguments: the receiver tests whether the output is in a typical set corresponding to the tested word. This occurs with high probability if the tested word was actually sent. If two input sequences are very close, their corresponding typical sets are too similar, making them indistinguishable. To ensure distinguishability between inputs, a minimal pairwise distance of order $\sim n^{\frac14}$ is required. In other words, the sequences distinguishable by the typicality test form a packing in the input space. Optimising this approach yields linearithmic DI rates, but only up to the prefactor $\frac14$ in the Gaussian case.

In this framework, an error can occur when the noise vector has atypically large or small norm. Then, the output falls outside the typical set of the sent word $u$ (generating a missed identification if the receiver tests $u$), and it can fall inside a different typical set corresponding to a word $v$ (and that will generate a false identification if the receiver tests for $v$). A false identification can also happen even with typical error behaviour due to the fact that the typical sets corresponding to different input words have small overlaps.

\medskip

\noindent\textbf{Projective constructions:}
A different principle was introduced in \cite{galaxy-codes} leading to an improved lower bound of $\frac38$. The following phenomenon of random Gaussian vectors in high dimensions was used: with high probability, a random Gaussian noise vector is almost orthogonal to any fixed direction.

%Although this property may seem counter-intuitive in low dimensions, it becomes natural in high dimensions. A random vector whose energy is spread across many coordinates typically has small projection onto any specific direction. Equivalently, while the total noise energy may be moderate, its component along a prescribed direction is typically much smaller.

This observation suggests replacing typicality tests, which depend on the total size of the noise vector, by \emph{projective tests}, which care only about the contributions onto selected directions. By enforcing sufficient separation between the projections of different code words onto these directions, one can distinguish them with high probability, because the contribution of the added noise onto the projected direction is likely small:

\begin{proposition}\label{prop:projection}
    Let $\vec{u}\in\RR^n$ be a vector in a high dimensional space and $\vec{Z}=(Z_1,\dots,Z_n)$ a random noise vector where $Z_1,\dots,Z_n\sim\cN(0,1)$ are i.i.d. standard normal random variables. The projection of the noise vector onto  $\vec{u}$ satisfies:
    \[
    \Pr\left(\|\Pi_{\vec{u}}\vec{Z}\|\geq x\right)=2\Phi(-x)\leq\frac{1}{\sqrt{2\pi}}\frac{1}{x}e^{-x^2/2}.
    \]
\end{proposition}

In this setting, the errors occur when the noise contribution onto a selected direction is too large, causing the projection of the output sequence to be far away from the projection of the sent word, beyond a certain acceptance threshold.

The model proposed in \cite{galaxy-codes} applies the property described above recursively. To accept a word, its orthogonal projection onto several carefully chosen directions must be simultaneously small. More precisely, they introduce a fractal, galaxy-like construction: points are distributed on the surface of a sphere; each of these points serves as the centre of a smaller sphere containing a new layer of points; these, in turn, act as centres of yet smaller spheres, and so on. The elements of the innermost layer are the code words. This concrete hierarchical structure allows one to select particular projection directions for each word such that they satisfy pairwise strong separation properties in projection and can therefore be reliably distinguished. A combination of this new construction and the earlier typicality-based one achieves a linearithmic rate $\dot R=\frac38$, improving the previously known $\dot R=\frac14$ when only typicality-based decoding was used.

We begin in Section~\ref{sec:pre-results} by showing that a significantly simpler projective construction can already achieve the same linearithmic rate $\dot R=\frac38$. Building on the geometric intuition developed from this example, we then demonstrate in Section~\ref{sec:main_results} that an optimised geometric construction can in fact attain a linearithmic rate of $\frac{1}{2}$, thereby closing the gap for DI over Gaussian channels. The results are completed in Section~\ref{sec:RR_tradeoff} with a study of the tradeoff between rates and errors. We conclude with a discussion in Section~\ref{sec:conclusions}.

%%%%%%%%%%%%%%%%%%%%%%%%
%%%%%%%%%%%%%%%%%%%%%%%%
%%%%%%%%%%%%%%%%%%%%%%%%
%%%%%%%%%%%%%%%%%%%%%%%%

\section{Preliminary results}\label{sec:pre-results}
We provide here the necessary tools and intuition needed to tackle the main results in the subsequent sections. We start with some geometric considerations in Section~\ref{ssec:geometry} that will ease the subsequent calculations. Section~\ref{ssec:single-layer} contains the initial results of the paper: we show that a single-layer construction, much simpler than the multi-layer galaxy construction in \cite{galaxy-codes}, can achieve the same rate and with better error bounds. Finally, in Section~\ref{ssec:intuition}, we present a key observation and result that suggests how to improve the rates further.

\subsection{Geometric considerations}\label{ssec:geometry}
We include in this section several geometric notions that will play a central role in the code constructions. In particular, we consider arrangements of points on the surface of a sphere satisfying a minimum separation condition. While this is equivalent to the classical notion of a packing on the sphere surface, we formulate it in terms of angular separation and projected distance, which arise more naturally in our constructions.

\begin{definition}\label{def:angle-dense}
Given $d>0$, an arrangement of points $\{u_j\}_{j=1}^{N}$ on the surface of a ball $\cB_{\vec o}(r,n)$ is called $d$-\emph{angle-dense} if for every pair of distinct points $u_j\neq u_k$ it holds that
\[
\|\Pi_{\vec{u}_k}\vec{u}_j-\vec{u}_k\|\geq d .
\]
\end{definition}

\begin{lemma}\label{lemma:angle}
The minimum separation angle $\theta$ between any two different points $u_j\neq u_k$ in a $d$-angle-dense arrangement on the surface of the $r$-radius ball $\cB_{\vec o}(r,n)$ satisfies
\[
\sin\frac{\theta}{2} \geq \sqrt{\frac{d}{2r}} .
\]
\end{lemma}
\begin{proof}
Using basic trigonometry and referring to Figure~\ref{fig:AngleDense} for visual intuition, it is clear that in a $d$-angle-dense construction
\begin{equation}
  r\cos\theta= r-\|\Pi_{\vec{u}_k}\vec{u}_j-\vec{u}_k\| \leq r-d,
\end{equation}
the inequality following from Definition \ref{def:angle-dense}. The terms can be rearranged as $\cos\theta\leq 1-\frac{d}{r}$, and by introducing the sine of half the angle $\theta$, we can complete the proof:
\[
\sin\frac{\theta}{2}=\sqrt{\frac{1-\cos\theta}{2}}\geq\sqrt{\frac{d}{2r}}.\qedhere
\]

\end{proof}

\begin{figure}
    \centering
    \includegraphics[width=0.98\linewidth]{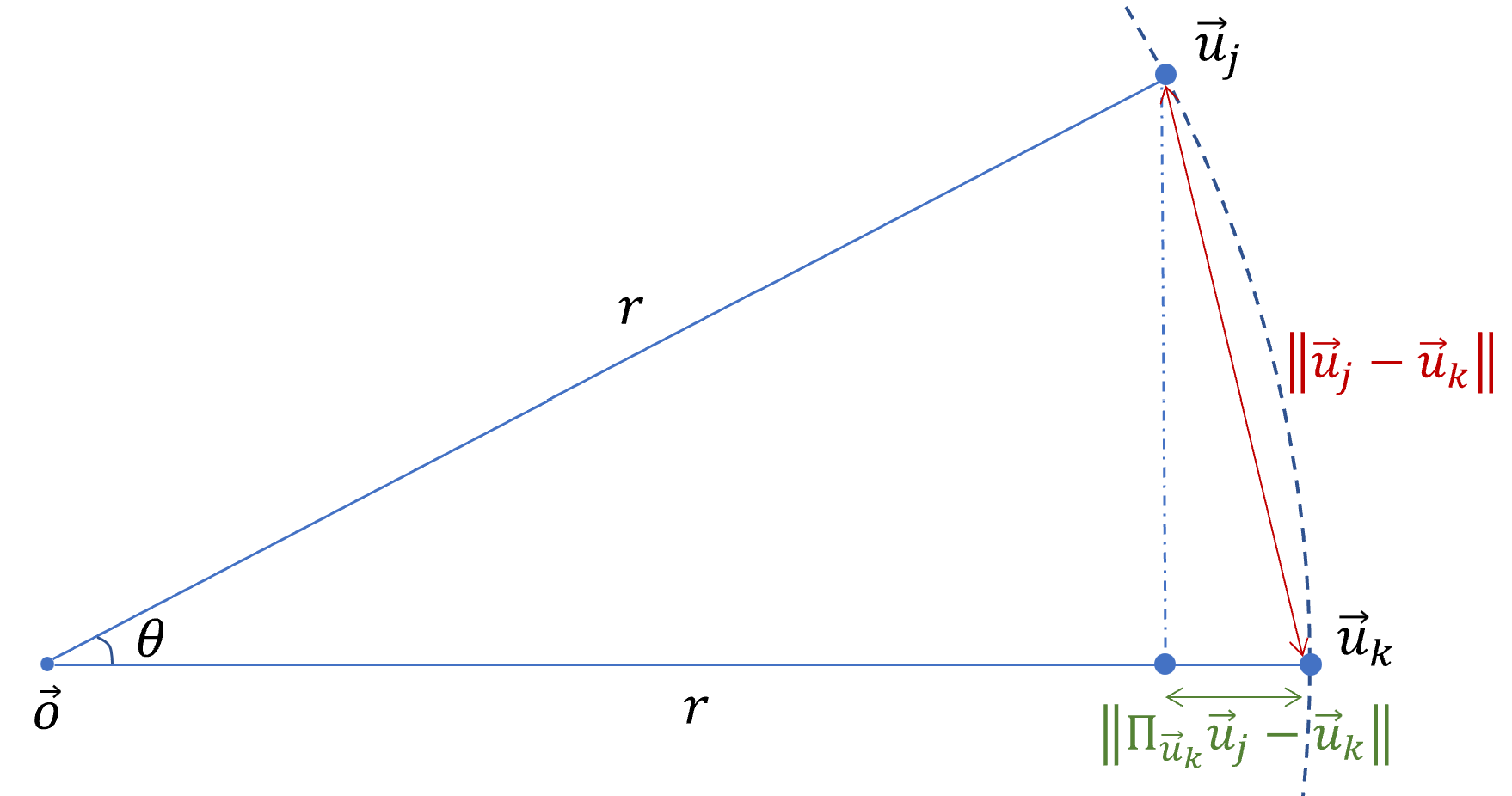}
    \caption{2D visualization of the projection of one vector on the surface of a sphere onto another neighbouring one.}
    \label{fig:AngleDense}
\end{figure}

The minimum distance between any two points (which would classically define the packing on the surface of the sphere) is given by:
\begin{lemma}\label{lemma:distance}
The Euclidean distance between any pair of distinct points $u_j\neq u_k$ in a $d$-angle-dense arrangement is
\[
\|\vec{u}_j-\vec{u}_k\|\geq\sqrt{2rd}.
\]
\end{lemma}
\begin{proof}
The distance $\|\vec{u}_j-\vec{u}_k\|$ is exactly the chord between the two points on the surface of the sphere and can therefore be calculated as 
\begin{equation}\label{eq:chord}
\|\vec{u}_j-\vec{u}_k\|=2r\sin\frac{\theta}{2}\geq\sqrt{2rd},
\end{equation}
where we have used Lemma \ref{lemma:angle} in the last inequality, completing the proof.
\end{proof}

\begin{proposition}\label{prop:size_bound}
The number $N$ of elements in a $d$-angle-dense arrangement on a ball $\cB_{\vec o}(r,n)$ of radius $r$ and dimension $n$ satisfies
\[
\log N \geq \frac{n}{2}\log\left(\frac{2r}{d}\right).
\]
\end{proposition}
\begin{proof}
It is well known (see, e.g., \cite{packing_book}) that the maximal number $N$ of points on the surface of an $n$-dimensional sphere with minimum angular separation $\theta$ satisfies
\begin{equation}
N \geq \left(\sin\frac{\theta}{2}\right)^{-n[1+o(1)]}.
\end{equation}
Taking the logarithm of this bound and combining it with Lemma~\ref{lemma:angle} immediately yields:
\begin{equation}
\begin{split}
\log N &\geq n[1+o(1)]\log\left(\frac{1}{\sin(\theta/2)}\right)\\
&\geq \frac{n}{2}\log\left(\frac{2r}{d}\right),
\end{split}
\end{equation}
completing the proof.
\end{proof}

\subsection{Single-layer construction}\label{ssec:single-layer}
This section presents a simple single-layer angle-dense construction that reproduces the rate $\frac{3}{8}$ for the linearithmic capacity using a significantly simpler design. This rate was recently achieved in \cite{galaxy-codes} through intricate multi-layer galaxy constructions placed inside the input packing balls required for typicality decoding in the original scheme \cite{DI-fading,RRGauss-arXiv}. Here we show that a single-layer angle-dense arrangement inside each ball already suffices to achieve the rate $\dot R=\frac{3}{8}$.

This section also aims to warm up the reader for the subsequent results, which require more involved constructions but rely on the same core ideas. We begin with a preliminary calculation of the rate obtained from the simplest possible single-layer projection-based code, without any use of typicality ideas.

\begin{theorem}\label{thm:single_layer}
Let $\lambda:=2\Phi(-\log n)$. A good $(n,N,\lambda,\lambda)$ DI code can be constructed, with the words forming an angle-dense arrangement, achieving a linearithmic rate $\dot{R}=\frac14$.
\end{theorem}
\begin{proof}
Let the code words be a $d$-angle-dense arrangement of points $\{u_i\}_{i=1}^N$ on the surface of the input space, a ball of radius $r=\sqrt{Pn}$ with $P>0$, and choose $d=2\sigma\log n$ to be the minimum projected distance. Then, by Proposition \ref{prop:size_bound}, we have 
\begin{equation}\label{eq:single-layer-size}
\log N\geq\frac{n}{2}\log\left(\frac{\sqrt{Pn}}{\sigma\log n}\right).
\end{equation}
For each $u_i$, the following set is chosen as identification effect:
\begin{equation}\label{eq:dec_single_layer}
\cD_i=\{y\in\RR^n:\|\Pi_{\vec{u}_i}\vec{y}-\vec{u}_i\|\leq \sigma\log n\}.
\end{equation}

The family $\{(u_i,\cD_i):i\in[N]\}$ will be a good DI code if the error probabilities introduced in Definition \ref{def:DI_code} are bounded. The probability of correctly identifying the sent word $u_i$ is given by
\begin{equation}\label{eq:Proj_error1}
\begin{split}
\Pr&(\cD_i|\vec{u}_i)=\Pr\left(\|\Pi_{\vec{u}_i}\vec{Y}-\vec{u}_i\|\leq\sigma\log n|\vec{u}_i\right)\\
&=\Pr\left(\|\Pi_{\vec{u}_i}(\vec{u}_i+\sigma\vec{Z})-\vec{u}_i\|\leq\sigma\log n\right)\\
&=\Pr\left(\|\Pi_{\vec{u}_i}(\vec{u}_i)-\vec{u}_i+\sigma \Pi_{\vec{u}_i}(\vec{Z})\|\leq\sigma\log n\right)\\
&\stackrel{(a)}{=}\Pr\left(\|\Pi_{\vec{u}_i}(\vec{Z})\|\leq\log n\right)\\
&\stackrel{(b)}{=}1-2\Phi(-\log n)=:1-\lambda,
\end{split}
\end{equation}
where in $(a)$ the fact $\Pi_{\vec{u}_i}\vec{u}_i=\vec{u}_i$ is used and $\sigma$ simplified on both sides of the inequality, and $(b)$ follows from Proposition~\ref{prop:projection}. Hence the probability of missed identification can be bounded by $P_{e,1}(\vec{u}_i)\leq \lambda$.

The probability of falsely identifying a word $\vec u_i$ when a different $u_j$ is sent can be similarly bounded for any pair of words $\vec{u}_i\neq\vec{u}_j$ as follows:
\begin{equation}\label{eq:Proj_error2}
\begin{split}
\Pr&(\cD_i|\vec{u}_j)=\Pr\left(\|\Pi_{\vec{u}_i}\vec{Y}-\vec{u}_i\|\leq\sigma\log n|\vec{u}_j\right)\\
&=\Pr\left(\|\Pi_{\vec{u}_i}(\vec{u}_j+\sigma\vec{Z})-\vec{u}_i\|\leq\sigma\log n\right)\\
&\stackrel{(c)}{\leq}\Pr\left(\|\Pi_{\vec{u}_i}\vec{u}_j-\vec{u}_i\|-\sigma\|\Pi_{\vec{u}_i}\vec{Z}\|\leq\sigma\log n\right)\\
&\stackrel{(d)}{\leq}\Pr\left(\|\Pi_{\vec{u}_i}\vec{Z}\|\geq\log n\right)\\
&=2\Phi(-\log n)=:\lambda,
\end{split}
\end{equation}
where the reverse triangle inequality has been applied in $(c)$, and $(d)$ follows from $\|\Pi_{\vec{u}_i}\vec{u}_j-\vec{u}_i\|\geq d:=2\sigma\log n$. Hence, the probability of false identification can also be bounded for any pair $\vec{u}_i\neq\vec{u}_j$ by 
\(P_{e,2}(\vec{u}_j|\vec{u}_i)\leq\lambda:=2\Phi(-\log n)\).
Now, using Proposition \ref{prop:projection} we observe
\begin{equation}
\lambda:=2\Phi(-\log n)\leq \sqrt{\frac{2}{\pi}}\frac{1}{\log n}e^{-(\log n)^2/2}.
\end{equation}
Therefore, both errors vanish as $n\to\infty$, and the construction defines a good DI code of size $N$. 

It is only left to calculate the linearithmic rate. Using Equation~\eqref{eq:single-layer-size}:
\begin{equation}\label{eq:R_example1}
\dot{R}(n):=\frac{\log N}{n\log n}\geq\frac{1}{4}-\frac12\frac{\log\log n}{\log n}.
\end{equation}
Notice that, as the second term vanishes as $n\to\infty$, for any arbitrarily small constant $\epsilon>0$ there must exist an $n_0$ such that we have $\dot R(n)\geq \frac14-\epsilon$ for all $n\ge n_0$. Therefore, $\dot{R}=\frac14$ is an achievable rate, and the proof is complete.
\end{proof}

While the code construction above does not achieve the performance of the galaxy construction in \cite{galaxy-codes} (a rate of $\frac38$), it already reproduces the linearithmic performance of the best known typicality-based codes from \cite{DI-fading,RRGauss-arXiv} with an remarkably simple and purely geometric construction which avoids any typicality arguments. 

\medskip
We show next that combining this single-layer projection-based code with a good typicality-based code recovers the rate $\dot R=\frac38$. The idea is to first apply a typicality decoder to identify a particular ball, and then use a projective decoder to identify a sequence within an arrangement inside that ball, mimicking the strategy of \cite{galaxy-codes} but with a substantially simpler construction (a single layer inside each ball rather than multiple layers).

Let $(n,N^{(\cT)},\lambda_1^{(\cT)},\lambda_2^{(\cT)})$ be a good typicality-based code as constructed in \cite{DI-fading}. Such a code can be obtained, for some constant $\alpha>0$ and arbitrarily small $b>0$, by packing balls $\cB_{i}(\alpha n^{\frac{1-b}{4}},n)$ in the input space ($i\in[N^{(\cT)}]$), and satisfies 
\begin{equation}\label{eq:typical_size}
    \log N^{(\cT)}\geq\frac{n}{4}(1-b)\log n-O(n),
\end{equation}
which shows that the rate $\dot R=\frac{1}{4}$ is achievable.

\begin{theorem}\label{thm:reproducing_3on8}
Given $\lambda=2\Phi(-\log n)$ and a good typicality-based $(n,N^{(\cT)},\lambda_1^{(\cT)},\lambda_2^{(\cT)})$ DI code (as described above), we can create another $(n,N,\lambda+\lambda_1^{(\cT)},\lambda+\lambda_2^{(\cT)})$ DI
code combining it with an angle-dense single-layer construction inside each packing ball, achieving a linearithmic rate $\dot{R}=\frac38$.
\end{theorem}
\begin{proof}
Place a single-layer angle-dense construction of radius $r^{(\cT)}=\alpha n^\frac{1-b}{4}$ inside each packing ball $\cB_i$. It has $N_i$ elements at a minimum projective distance $d=2\sigma\log n$ which can be identified with the decoder $\cD_i$ in Equation~\eqref{eq:dec_single_layer}, with errors bounded by $\lambda=2\Phi(-\log n)$. Indeed, notice that, maintaining the choice of $d$, the change on the radius does not affect the error bounds in Equations \eqref{eq:Proj_error1} and \eqref{eq:Proj_error2}. 

All the words from the same single-layer arrangement (inside the same ball of radius $r^{(\cT)}$) need to be simultaneously identifiable by a typical decoder. This can be ensured by doubling the radii of the decoding balls $\cB_{i}(2\alpha n^\frac{1-b}{4},n)$ because, then, no word from a different angle-dense system can be at a distance smaller than $r^{(\cT)}$ from the tested word. Notice also that doubling the radii does not affect the scaling of $N^{(\cT)}$ at first order, as the packing number only changes by a constant factor which is immediately absorbed by the lower order term $O(n)$ in Equation~\eqref{eq:typical_size}.

Invoking now Proposition \ref{prop:size_bound}, one observes that each single-layer arrangement satisfies:
\begin{equation}
\log N_i\geq\frac n2\log\left(\frac{\alpha}{\sigma}\frac{n^{\frac{1-b}{4}}}{\log n}\right).
\end{equation}

The construction has $N_i$ words inside each of the $N^{(\cT)}$ balls, so the total number of elements $N=N_i\cdot N^{(\cT)}$ and its logarithm is bounded by: 
\begin{equation}
\begin{split}
\!\!\log N&= \log N_i +\log N^{(\cT)}\\
&\geq\frac{n}{2}\log\left(\frac{\alpha}{\sigma}\frac{n^{\frac{1-b}{4}}}{\log n}\right)\!+\!\frac{n}{4}(1-b)\log n-O(n)\\
&\geq\left(\frac{n}{8}+\frac{n}{4}\right)(1-b)\log n-O(n\log\log n).
\end{split}
\end{equation}
Dividing by $n\log n$ we obtain the following linearithmic rate: 
\begin{equation}
    \dot{R}(n)=\frac38(1-b)-O\left(\frac{\log\log n}{\log n}\right).
\end{equation}
Similarly to the last proof, as one can choose an arbitrarily small $b>0$ and take 
large enough values of $n$ such that $\dot R(n)\geq\frac{3}{8}-\epsilon$ for any $\epsilon>0$, it can be concluded that $\dot R=\frac38$ is an achievable rate.

The error analysis is immediate because the total errors of missed and false identification can be bounded by the sum of the typicality decoder and the projective decoder errors: $\lambda_1=\lambda_1^{(\cT)}+\lambda$ and $\lambda_2=\lambda_2^{(\cT)}+\lambda$, which vanish as $n\to\infty$. The proof is completed.
\end{proof}

We have achieved the same rate performance as in \cite{galaxy-codes} but with a much simpler construction: placing inside the typicality-decoding balls an angle-dense single-layer code instead of a full multi-layer galaxy. Moreover, in the multi-layer galaxy construction the missed identification error accumulates over the $L=O(\log n)$ layers: $\lambda'_1=\lambda_1^{(\cT)}+L\lambda$, whereas in the present single-layer scheme the corresponding bound is smaller: $\lambda_1=\lambda_1^{(\cT)}+\lambda$.

\subsection{Double-layer codes and motivation for multi-layer cases}
The single-layer construction above is angle-dense, meaning that no extra sequences can be added on the surface of the sphere with the necessary projective distance for identification using our projective decoder. However, notice the following two key facts. First, in spite of this density, the distance between elements, bounded in Lemma~\ref{lemma:distance} as 
\(
\|\vec{u}_j-\vec{u}_k\|\geq\sqrt{2rd},
\) 
is still large. Indeed, in the setting of Theorem~\ref{thm:single_layer}, where $r=\sqrt{nP}$ and $d=2\sigma\log n$, we have 
\begin{equation}\label{eq:intuition_ball_size}
    \|\vec{u}_j-\vec{u}_k\|\geq\sqrt{4\sigma P\sqrt{n}\log n}=O\left(n^{\frac14}\sqrt{\log n}\right).
\end{equation}
Actually, it is slightly (logarithmically) larger than the distance between elements needed for typicality decoding,  where we required a packing of balls with radius $r^{(\cT)}=O(n^\frac{1-b}{4})$.
\footnote{The fact that we obtain the same rate using only the typicality construction or only the single-layer projective decoding construction in Theorem~\ref{thm:single_layer} is clearer now: both constructions result from a packing of similar size. In the projective case, the packing is only done on the surface of the input space, while in the typicality case the packing is done in the whole input. But the number of packing elements is not affected at first order because the volume of a sphere concentrates on its surface, leading to the same achievable rate.}

Second, while the size $r^{(\cT)}$ of the packing balls is a necessary requirement of the typicality construction such that the errors can vanish, we can have good projection-based codes with words much closer together. Indeed, in Theorem~\ref{thm:reproducing_3on8}, we have constructed a good projective code inside each typicality packing ball $\cB_i(2\alpha n^\frac{1-b}{4},n)$. Using again Lemma~\ref{lemma:distance} we observe that the pairwise distance between the words in each of those projective codes is
\begin{equation}\label{eq:intuition2}
    \|\vec{u}_j^{(i)}-\vec{u}_k^{(i)}\|\geq\sqrt{8\sigma \alpha n^\frac{1-b}{4}\log n}=O\left(n^\frac{1-b}{8}\sqrt{\log n}\right),
\end{equation}
much smaller than the distance between the elements of the previously analysed codes, cf.~Equation~\eqref{eq:intuition_ball_size}. Nevertheless, we observed that these words defined good codes and had a positive impact of $\frac18$ on the achievable rate. 

Then, if the distance between words in the single-layer construction is similar to the one in the typicality-based code, and another projection-based code allows to distinguish arrangements of words that are much closer together, 
%Given the similarity between the single-layer projection-based code on the surface of the input space and the typicality-based code in previous approaches, exposed in the first remark above; and the existence of a smaller projective code with positive contribution to the rate, described in the second remark; the following question arises naturally: 
is it possible to construct a double-layer projective code that achieves the linearithmic rate $\frac38$ avoiding the use of typicality decoding? If it is, can we keep adding smaller layers successively with positive contributions to the achievable rate, bringing us closer to the known upper bound $\dot R\leq\frac12$?

The main difficulty is ensuring that the two codes can be simultaneously good. In Theorem~\ref{thm:reproducing_3on8}, the typicality decoder guaranteed that no element in a particular projective code (inside a typicality ball) could be confused with any element of another arrangement (in a different ball), so we could automatically apply both decoders successively. We have to be a bit more careful now.
A naive construction of a double-layer code would be to place around each of the elements $\{\vec o_{s_1}\}_{s_1=1}^{N_1}$ of a $d$-angle-dense arrangement in a ball of radius $r$ another dense system $\{\vec o_{s_1,s_2}\}_{s_2=1}^{N_2}$ in a smaller ball of radius $\sqrt{r}$ (the radii hierarchy extracted from the observations above, cf.~Equations~\ref{eq:intuition_ball_size} and \eqref{eq:intuition2}). The second layer elements $\{\vec o_{s_1,s_2}\}_{s_1,s_2=1}^{N_1,N_2}$ would be the code words, see Figure~\ref{fig:no_subspaces}. 

Let the word $\vec{u}_i=\vec o_{s_1,s_2}$ be sent. Then, for the decoding, one would use a first layer projective decoder $\cD_i^{(1)}$ to identify the first layer element $\vec{o}_{s_1}$ around which the word is placed, and a second layer decoder  $\cD_i^{(2)}$ to identify the word. Let
\begin{equation}\label{eq:double-layers}
\begin{split}
\cD_i^{(1)}&=\{y\in\RR^n:\|\Pi_{s_0\to s_1}\vec{y}-\vec{o}_{s_1}\|\leq \sigma\log n\},\\
\cD_i^{(2)}&=\{y\in\RR^n:\|\Pi_{s_1\to s_2}\vec{y}-\vec{o}_{s_1,s_2}\|\leq \sigma\log n\};
\end{split}
\end{equation}
where $\Pi_{s_0\to s_1}$ is the orthogonal projection onto the vector $\vec v_{s_0\to s_1}=\vec o_{s_1}-\vec o_{s_0}$, with $\vec o_{s_0}=\vec o$ the origin, and similarly $\Pi_{s_1\to s_2}$ is the projection onto $\vec v_{s_1\to s_2}=\vec o_{s_1,s_2}-\vec o_{s_1}$. We call these vectors onto which we project \emph{path vectors}. 

\begin{figure}
    \centering
    \includegraphics[width=0.95\linewidth]{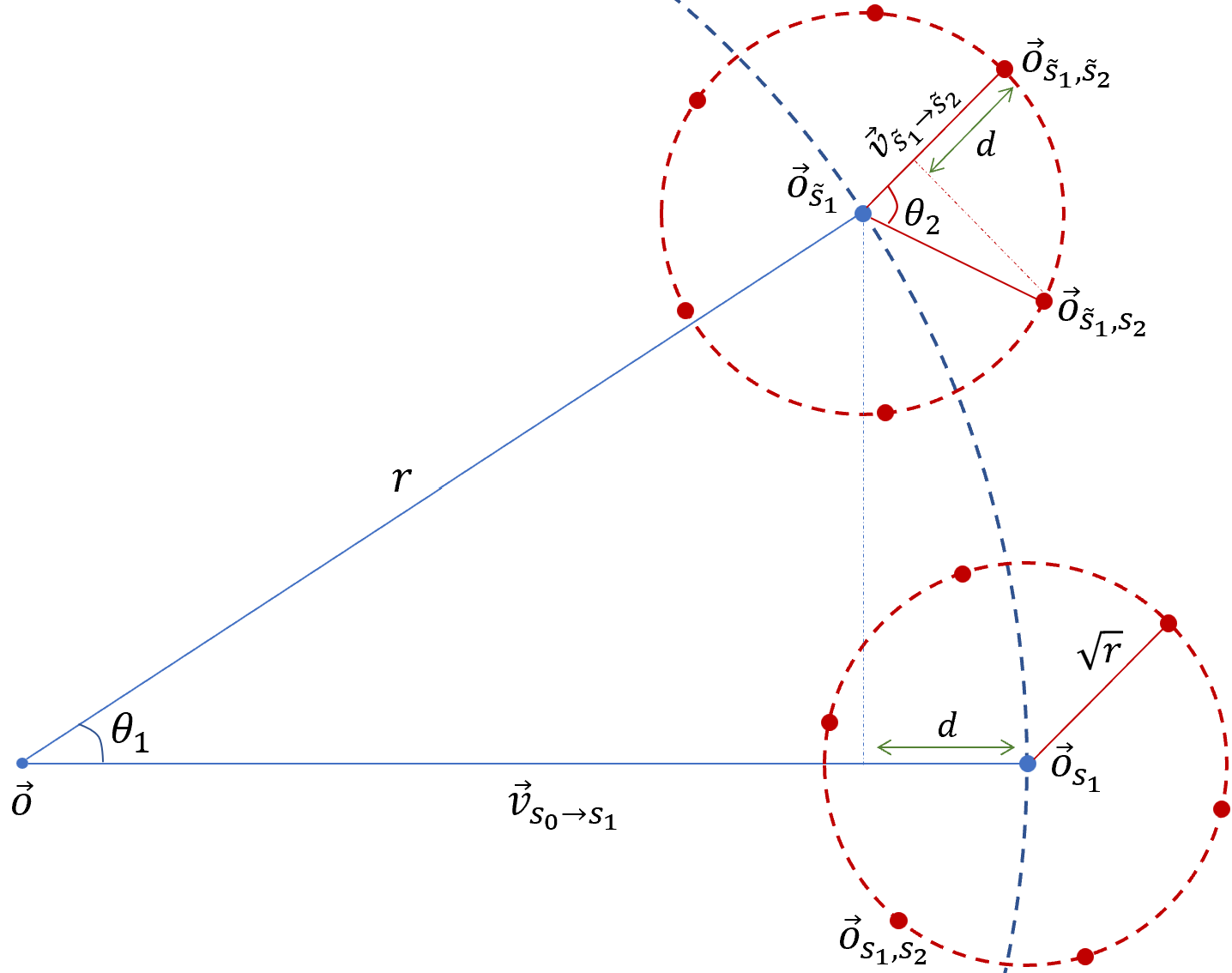}
    \caption{A 2D representation of two neighbouring first layer elements $\vec{o}_{s_1}$ and $\vec{o}_{\tilde{s}_1}$ in a naive code with two layers. As both first (blue) and second layer (red) arrangements are $d$-angle-dense, we observe how the orthogonal projection of $\vec{o}_{\tilde{s}_1}$ onto $\vec{v}_{s_0\to s_1}$ is at a distance $d$ from $\vec{o}_{s_1}$, and similarly the projection of $\vec{o}_{\tilde{s}_1,s_2}$ onto $\vec{v}_{\tilde{s}_1\to \tilde s_2}$ is also at a distance $d$ from $\vec{o}_{\tilde{s}_1, \tilde s_2}$.}
    \label{fig:no_subspaces}
\end{figure}

However, looking at Figure~\ref{fig:no_subspaces}, notice that if one tries to identify a particular second layer system (say, the one around $\vec o_{s_1}$) the projection of code words onto the path vector $\vec{v}_{s_0\to s_1}$ can be very far away from $\vec o_{s_1}$. In fact, as far away as the size of the second layer radius, which is much larger [$O(\text{poly}\,n)$] than the acceptance range of the decoder [$O(\log n)$]. Also, some code words from a different second layer system (take $\vec o_{\tilde s_1,\tilde s_2}$ in the figure, for example) have a projection onto $\vec{v}_{s_0\to s_1}$ very close to $\vec o_{s_1}$. Then, the proposed first layer decoder cannot identify correctly all the elements in the particular second layer arrangement, and can also wrongly accept elements from different second layer systems. 

%One could try to fix these problems by increasing the acceptance range of the decoders to something larger than the second layer radius. But then, remembering the proof of Theorem~\ref{thm:single_layer}, one would be required to similarly increase the value of the minimum projected distance $d$ that defines the geometry of the code in order to correctly bound the second type of error. However, this massive increase on $d$ (from logarithmic to polynomial in $n$) would critically affect the code size. Indeed, by looking at Proposition~\ref{prop:size_bound}, one observes that forcing $d=O(r)$ implies $\log N=O(n)$, so the linearithmic scaling of the message length is lost.

%Another fix could be changing the radius hierarchy and try to reuse the ideas above. However, a multi-layer optimisation of this approach yields exactly the construction in \cite{galaxy-codes}, which cannot produce further improvements on the rate, and is no better than the single-layer angle-dense approach proposed in Theorem~\ref{thm:reproducing_3on8}.

%\medskip

There is, however, a trick to bypass the projective issues of multiple angle-dense layers, which consists in constructing the second layer arrangements in particular subspaces. Specifically, around each first layer element $s_1$, the second layer arrangement is a $d$-angle-dense construction in the ball $\cB_{s_1}(\sqrt{r},n-1)\in \vec{v}_{s_{0}\to s_{1}}^\perp$, the subspace orthogonal to $\vec{v}_{s_{0}\to s_{1}}$:
\begin{equation}
\vec{v}_{s_{0}\to s_{1}}^\perp:=\{\vec{x}\in\RR^n:\vec{x}\cdot\vec{v}_{s_{0}\to s_{1}}=0\}.
\end{equation}
With this projection we lose a dimension in all second layer arrangements, but we gain the following fundamental property (see Figure~\ref{fig:subspaces} for visual intuition):

\begin{figure}
\centering
\includegraphics[width=0.95\linewidth]{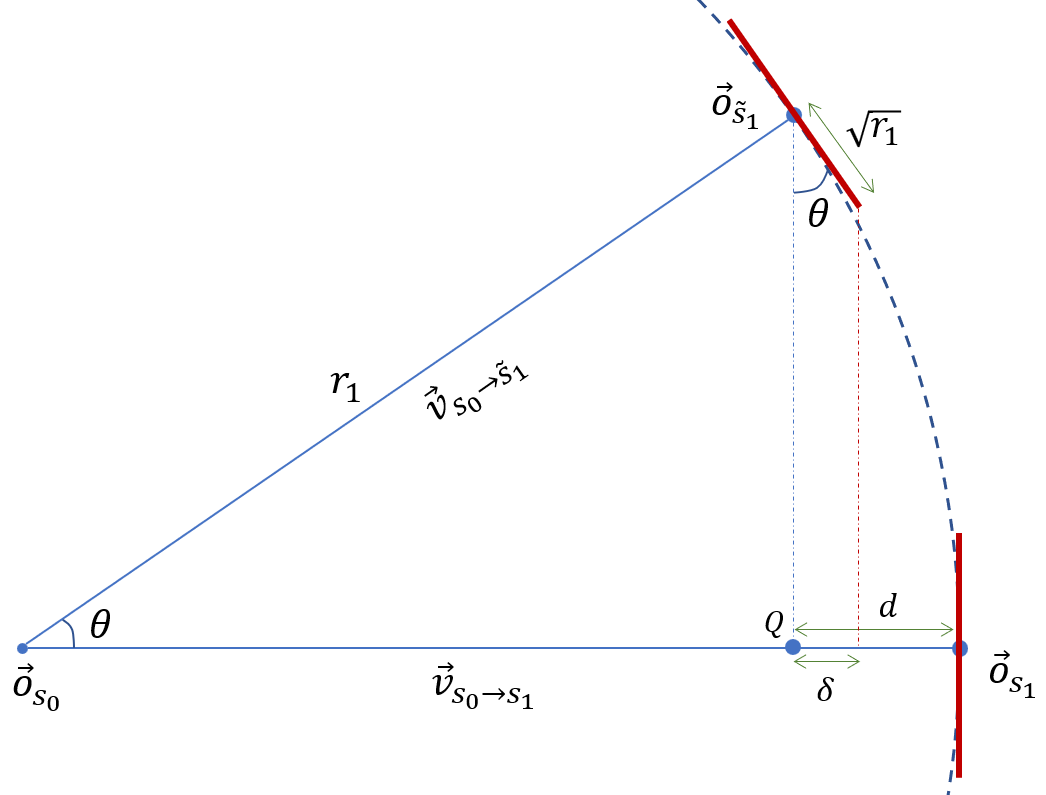}
\caption{Visual representation of the plane containing the origin $\vec{o}=\vec{o}_{s_0}$, and two first layer elements $\vec{o}_{\tilde s_1}$, and $\vec{o}_{s_1}$. The second layer systems in this plane are the red lines (orthogonal to the corresponding path vectors); thus, the projection of any second layer element around $\vec{o}_{\tilde s_1}$ onto the path vector $\vec v_{s_0\to s_1}$ is at most at a distance $\delta$ from $Q:=\Pi_{s_0\to s_1}\vec o_{\tilde s_1}$. Also, the projection of any second layer element around $\vec{o}_{s_1}$ onto $\vec v_{s_0\to s_1}$ is $\vec{o}_{s_1}$.}
\label{fig:subspaces}
\end{figure}

\begin{lemma}\label{lemma:2-layer-delta}
In the double-layer construction above, the distance $\delta$ between the projection of the first layer element $\vec o_{\tilde s_1}\in\cB_{s_0}(r,n),$ and the projection of any second layer element $\vec o_{\tilde s_1,\tilde s_2}\in\cB_{s_1}(\sqrt{r},n-1)$ around it onto a certain path vector $\vec v_{s_0\to s_1}$ is bounded by:
\[
\delta:=\|\Pi_{s_0\to s_1}\vec o_{\tilde s_1} - \Pi_{s_0\to s_1}\vec o_{\tilde s_1,\tilde s_2}\|<\sqrt{2d}
\]
\end{lemma}
\begin{proof}
It is just necessary to notice that $\delta\leq\sqrt{r}\sin\theta$ (refer to Figure~\ref{fig:subspaces}). Then, with $\cos\theta\leq1-\frac dr$,
\begin{equation*}
\delta<\sqrt{r}\sin\theta=\sqrt{r(1-\cos\theta^2)}=
\sqrt{2d-\frac{d^2}{r}}<\sqrt{2d}.\qedhere
\end{equation*}
\end{proof}

Therefore, if one chooses $d$ much smaller than the radii, the projections of different second layer systems concentrate tightly around the projection of their corresponding centres, which can be identified with our projective decoders! 

\begin{theorem}\label{thm:double-layer}
Let $\lambda:=2\Phi(-\log n)$. A good $(n,N,2\lambda,\lambda)$ DI code can be constructed with two angle-dense layer arrangements achieving a linearithmic rate $\dot{R}=\frac38$.
\end{theorem}
\begin{proof}
Choose $d=3\sigma\log n$ (this slight increase in $d$ will compensate $\delta$). Construct a $d$-angle-dense arrangement on the surface of $\cB_{s_0}(r_1,n)$ with $r_1=\sqrt{nc}$, $c<P$ and elements $s_1\in[N_1]$. Around each $s_1$, construct another $d$-angle-dense arrangement on the surface of $\cB_{s_1}(r_2,n-1)\subset\vec v_{s_0\to s_1}^\perp$, with $r_2=\sqrt{r_1}=\sqrt[4]{nc}$ and elements $s_2\in[N_2]$. The set of all second layer elements $\{\vec o_{s_1,s_2}\}_{s_1=1,s_2=1}^{N_1,N_2}$ is the codebook. 
%Both $N_1$ and $N_2$ are bounded by Proposition~\ref{prop:size_bound}. 

For the decoding strategy, take the projective decoders in Equation~\eqref{eq:double-layers}. Then, the decoder corresponding to the word $u_i$ is $\cD_i:=\cD_i^{(1)}\cap\cD_i^{(2)}$ for all $i\in[N]$. If the word sent and tested is the same, say $\vec u_i=\vec o_{s_1,s_2}$, then the probability that the first layer decoder identifies the received sequence as a word around the (correct) first layer element $\vec o_{s_1}$ is:
\begin{align}\label{eq:double-layers_error1_1}
\Pr&(\cD_{i}^{(1)}|\vec u_i)=\Pr\left(\|\Pi_{s_{0}\to s_{1}}\vec{Y}-\vec o_{s_1}\|\!\leq\!\sigma\log n|\vec o_{s_1,s_2}\right)\nonumber\\
&\!\!=\Pr\left(\|\Pi_{s_{0}\to s_{1}}\vec o_{s_1,s_2}-\vec o_{s_1}+\sigma \Pi_{s_{0}\to s_{1}}\vec{Z}\|\leq \sigma\log n\right)\nonumber\\
&\!\!\stackrel{(a)}{=}\Pr\left(\|\Pi_{s_{0}\to s_{1}}\vec{Z}\|\leq\log n\right)\\
&\!\!=1-2\Phi(-\log n)=1-\lambda\nonumber,
\end{align}
where in $(a)$ we have used that $\Pi_{s_{0}\to s_{1}}\vec o_{s_1,s_2}=\vec o_{s_1}$, because the ball $\cB_{s_1}(\sqrt[4]{nc},n-1)$ on the surface of which $\vec o_{s_1,s_2}$ lives is orthogonal to the path vector $\vec v_{s_0\to s_1}$ (see Figure~\ref{fig:subspaces}). The probability $\Pr(\cD_{i}^{(2)}|\vec u_i)$ that the second layer decoder identifies the correct code word is bounded similarly:
\begin{align}\label{eq:double-layers_error1_2}
\Pr&\left(\|\Pi_{s_{1}\to s_{2}}\vec o_{s_1,s_2}-\vec o_{s_1,s_2}+\sigma \Pi_{s_{1}\to s_{2}}\vec{Z}\|\leq \sigma\log n\right)\nonumber\\
&=\Pr\left(\|\Pi_{s_{1}\to s_{2}}\vec{Z}\|\leq\log n\right)\\
&=1-2\Phi(-\log n)=1-\lambda\nonumber,
\end{align}
as $\Pi_{s_{1}\to s_{2}}\vec o_{s_1,s_2}=\vec o_{s_1,s_2}$.
Therefore, by the union bound, the total probability of missed identification satisfies
\begin{equation}
P_{e,1}(\vec u_i)\leq \Pr(\cD_{i}^{(1)^c}|\vec u_i)+\Pr(\cD_{i}^{(2)^c}|\vec u_i)=2\lambda,
\end{equation}
where the superscript $c$ indicates the complementary set.

If the words tested $\vec u_i=\vec{o}_{s_1,s_2}$ and sent $\vec u_j=\vec{o}_{\tilde s_1,\tilde s_2}$ are in different second layer arrangements, i.e.~$s_1\neq\tilde s_1\in[N_1]$, the $\cD_i^{(1)}$ must have small probability of identifying the sent word:
\begin{align}
\Pr&(\cD_{i}^{(1)}|\vec{o}_{\tilde s_1,\tilde s_2})=\Pr\left(\|\Pi_{s_{0}\to s_{1}}\vec{Y}-\vec{o}_{s_1}\|\leq\sigma\log n|\vec{o}_{\tilde s_1,\tilde \tilde s_2}\right)\nonumber\\
&=\Pr\left(\|\Pi_{s_{0}\to s_{1}}\vec{o}_{\tilde s_1,\tilde s_2}-\vec{o}_{s_1}+\sigma \Pi_{s_{0}\to s_{1}}\vec{Z}\|\leq \sigma\log n\right)\nonumber\\
&\stackrel{(b)}{\leq}\Pr\left(\|\Pi_{s_{0}\to s_{1}}\vec{o}_{\tilde s_1,\tilde s_2}-\vec{o}_{s_1}\|-\sigma\|\Pi_{s_{0}\to s_{1}}\vec{Z}\|\leq\sigma\log n\right)\nonumber\\
&\stackrel{(c)}{\leq}\Pr\left(\|\Pi_{s_{0}\to s_{1}}\vec{Z}\|\geq\log n\right)\label{eq:double-layer_error_2_1}\\
&=2\Phi(-\log n)=\lambda,\nonumber
\end{align}
where $(b)$ follows from the triangle inequality, and $(c)$ from our choice of $d$ and Lemma~\ref{lemma:2-layer-delta}, indeed  
\begin{align}
\|\Pi_{s_{0}\to s_{1}}\vec{o}_{\tilde s_1,\tilde s_2}-\vec{o}_{s_1}\|=d-\delta&\geq3\sigma\log n-\sqrt{6\sigma\log n}\nonumber\\
&>2\sigma\log n.\label{eq:double-layer-delta-step}
\end{align}

Finally, if the words tested $\vec u_i=\vec{o}_{s_1,s_2}$ and sent $\vec u_k=\vec{o}_{s_1,\tilde s_2}$ are different elements in the same second layer arrangement, the second layer decoder $\cD_i^{(2)}$ must have small probability of identifying $\vec u_k$
\begin{align}
\Pr&(\cD_{i}^{(2)}|\vec{o}_{s_1,\tilde s_2})\!=\Pr\!\left(\!\|\Pi_{s_{1}\to s_{2}}\vec{Y}-\vec{o}_{s_1,s_2}\|\leq\sigma\log n|\vec{o}_{s_1,\tilde s_2}\!\right)\nonumber\\
&\!\!\leq\Pr\left(\|\Pi_{s_{1}\to s_{2}}\vec{o}_{s_1, \tilde s_2}-\vec{o}_{s_1,s_2}\|-\sigma\|\Pi_{s_{1}\to s_{2}}\vec{Z}\|\leq\sigma\log n\right)\nonumber\\
&\!\!\leq\Pr\left(\|\Pi_{s_{1}\to s_{2}}\vec{Z}\|\geq\log n\right)\label{eq:double-layer_error_2_2}\\
&\!\!=2\Phi(-\log n)=\lambda\nonumber.
\end{align}
As for any two different code words at least one of the two instances above must happen, we can conclude that the probability of false identification is no larger than $P_{e,2}(\vec u_i|\vec u_j)\leq\lambda$.

It is only left to count the size of the code. The total number of words (second layer elements) is $N=N_1\cdot N_2$. Then, by Proposition~\ref{prop:size_bound}, one finds
\begin{equation}
\begin{split}
\log N&=\log N_1+\log N_2\\
&\geq\frac{n}{2}\log\left(\frac{2r_1}{d}\right)+\frac{n-1}{2}\log\left(\frac{2r_2}{d}\right)\\
&=\left(\frac{n}{2}+\frac{n}{4}\right)\log r-O(n\log \log n)\\
&=\frac{3n}{8}\log(n)-O(n\log\log n).
\end{split}
\end{equation}
Thus, the rate is bounded by $\dot R(n)\geq\frac38-O(\log\log n/\log n)$. As the second term can be arbitrarily small for sufficiently large $n$, we conclude that $\dot R=\frac 38$ is an achievable rate.
\end{proof}

The trick of constructing the second layer in certain orthogonal subspaces allows the construction of a good code with two $d$-angle dense layers, each one with a positive contribution to the rate: the first layer contributed $\frac14$ and the second $\frac18$. This reproduces the best known rate results for DI over Gaussian channels with a purely geometric approach, without any use of typicality-based decoding ideas. Can we add more dense layers with linearithmic rate contributions?

%%%%%%%%%%%%%%%%%%%%%%%%
%%%%%%%%%%%%%%%%%%%%%%%%
%%%%%%%%%%%%%%%%%%%%%%%%
%%%%%%%%%%%%%%%%%%%%%%%%

\section{Closing the capacity gap}\label{sec:main_results}
We show in this section how the double-layer projection-based construction can be extended to an arbitrarily large number of layers. The resulting code can achieve a linearithmic rate of $\dot R=\frac12$ (see Theorem~\ref{thm:main}) closing, therefore, the gap between lower and upper bounds in the capacity of deterministic identification over Gaussian channels. Indeed, we conclude in Corollary~\ref{cor:capacity} that $\dot C_\text{DI}(\cG)=\frac12$. 

The unique geometrical nature of our code allows the construction of a capacity achieving universal code, in Section~\ref{ssec:universal}; that is, a code that does not depend on channel parameters (noise and power) upon encoding and decoding, and still attains the channel capacity.

\subsection{Capacity achieving code construction}
Let us construct a code of $L$ layers. Each one is denoted by $\ell\in[L]$, and it is a $d$-angle-dense arrangement of points on the surface of an $(n\!+\!1\!-\!\ell)$-dimensional ball embedded in $\RR^n$
\begin{equation}
\cB_{s^{\ell-1}}\left(r_\ell,n+1-\ell\right)\subset\left[\text{span}\left(\vec v_{s_0\to s_1},\dots,\vec{v}_{s_{\ell-2}\to s_{\ell-1}}\right)\right]^\perp,
\end{equation}
where $s^{\ell-1}=(s_1,\dots,s_{\ell-1})$ is the sequence of the elements, one at each layer, that navigates from the origin (to which we associate $s_0$) to a particular $\ell$-layer system; $r_\ell:=(cn)^{1/2^\ell}$ is the radius of each $\ell$-layer system (notice that $r_\ell=\sqrt{r_{\ell-1}}$, similarly to the double-layer case); the path vector that brings us from a centre $\vec o_{s^{\ell-1}}$ (uniquely associated to the sequence $s^{\ell-1}$) to $\vec o_{s^{\ell}}$ is denoted by $\vec v_{s_{\ell-1}\to s_\ell}:=\vec o_{s^{\ell}}-\vec o_{s^{\ell-1}}$ (notice that $\sum_{\ell=1}^L\vec v_{s_{\ell-1}\to s_\ell}=\vec{o}_{s^L}$); and $\left[\text{span}\left(\vec v_{s_0\to s_1},\dots,\vec{v}_{s_{\ell-2}\to s_{\ell-1}}\right)\right]^\perp$ is the subspace orthogonal to the corresponding path vectors at each layer (which are, in turn, mutually orthogonal between them).
\medskip

\noindent\textbf{Encoding:} The code words will be the $N$ elements of the innermost layer $\{\vec o_{s^L}\}_{s_1=1,\dots,s_L=1}^{N_1,\dots,N_L}$, with $N_\ell$ the number of elements in a particular $\ell$-layer system, which is bounded through Proposition~\ref{prop:size_bound} by:
\begin{equation}\label{eq:encoding_code_size}
\log N_\ell\geq\frac{n+1-\ell}{2}\log \left(\frac{2r_\ell}{d}\right),
\end{equation}
and we have $N=\prod_{\ell=1}^{L}N_{\ell}$. 
Notice that $\|\vec o_{s^L}\|^2=\sum_{\ell=1}^{L}(cn)^{2/2^\ell}=cn+o(n)<nP$, hence for $c<P$ and sufficiently large $n$ the power constraint is satisfied.

As all layers are $d$-angle-dense, we have that two different points $\vec o_{s^{\ell-1},s_\ell}\neq\vec o_{s^{\ell-1},\tilde s_\ell}$ in the arrangement around the same $\ell-1$ layer element $\vec o_{s^{\ell-1}}$ fulfil
\begin{equation}
\|\Pi_{s_{\ell-1}\to s_\ell}\vec o_{s^{\ell-1},\tilde s_\ell}-\vec o_{s^{\ell-1},s_\ell}\| \geq d.    
\end{equation}

\noindent\textbf{Decoding:} The identification effects at each layer $\ell$ corresponding to the particular code word $\vec u_i=\vec o_{s^L}$ are:
\begin{equation}\label{eq:decoder}
\cD_{i}^{(\ell)}=\{\vec{y}\in\RR^n:\|\Pi_{s_{\ell-1}\to s_{\ell}}\vec{y}-\vec{o}_{s^\ell}\|\leq\sigma\log n\},
\end{equation}
Then, the identification decoder corresponding to the message $i\in[N]$ encoded into the word $\vec o_{s^L}$ is $\cD_i=\bigcap_{\ell=1}^L\cD_{i}^{(\ell)}$. 

Similarly to what we saw for the double-layer code (cf.~ in Lemma~\ref{lemma:2-layer-delta}), we observe that for the multi-layer construction above, and thanks to the orthogonal subspace restriction at each layer, the projections of elements in inner layers concentrate around the projections of their centres. In particular:

\begin{lemma}\label{lemma:proj_dist}
Let $\vec o_{s^L}\neq\vec o_{\tilde s^L}$ be two different last-layer elements (code words) which differ at some layer $\ell$, that is, $s_\ell\in s^L$ is different from $\tilde s_\ell\in \tilde s^L$. Then,
\[
\|\Pi_{s_{\ell-1}\to s_\ell}\vec o_{\tilde s^L}-\Pi_{s_{\ell-1}\to s_\ell}\vec o_{s^L}\|> d-\sqrt{2Ld}=d-O(\sqrt{d}).
\]
\end{lemma}
\begin{proof}
First, notice that every last layer element $\vec o_{s^L}$ differs from its centre $\vec o_{s^\ell}$ at layer $\ell$ by the sum of path vectors $\sum_{j=\ell}^{L-1} \vec v_{s_j\to s_{j+1}}$, each of which is orthogonal to the rest and has norm $\|\vec v_{s_j\to s_{j+1}}\|=r_{j+1}$. Hence, by Pythagoras, $\|\vec o_{s^L}-\vec o_{s^\ell}\|^2:=\tilde r_{\ell+1}^2=\sum_{j=\ell}^{L-1} r_{j+1}^2$. Therefore, and looking at Fig.~\ref{fig:subspaces} for visual intuition, we have that 
\begin{equation}
\|\Pi_{s_{\ell-1}\to s_\ell}\vec o_{\tilde s^L}-\Pi_{s_{\ell-1}\to s_\ell}\vec o_{s^L}\|\geq d-\tilde r_{\ell+1}\sin\theta_\ell.
\end{equation}

Second, for any $\ell\geq 2$, observe that $\tilde r_{2}\sin\theta_1\geq\tilde r_{\ell+1}\sin\theta_\ell$. Additionally, $\tilde r_{2}=\sqrt{\sum_{\ell=2}^{L} r_{\ell}^2}<\sqrt{L}r_2$. Therefore, the separation $\delta$ between the projection of any last layer element and the projection of its centre at layer $\ell$ can always be bounded by:
\begin{equation}
\delta\leq\tilde r_2\sin\theta_1<\sqrt{Lr_1}\sqrt{1-\left(1-\frac{d}{r_1}\right)^2}<\sqrt{2Ld},
\end{equation}
where the second upper bound follows from the previous paragraph and the fact that $\cos \theta_1=1-\frac{d}{r_1}$, and the last bound from discarding the negative terms (which notice that vanish as $n$ increases).
Finally, $\|\Pi_{s_{\ell-1}\to s_\ell}\vec o_{\tilde s^L}-\Pi_{s_{\ell-1}\to s_\ell}\vec o_{s^L}\|\geq d-\delta>d-\sqrt{2Ld}$, and the proof is completed.
\end{proof}

%We have now all the necessary tools to prove the main result:
\begin{theorem}\label{thm:main}
Let $\lambda=2\Phi(-\log n)$, fix $L\in \NN$ a large natural number independent of $n$. Then, with the construction above, one can create a good angle-dense $L$-layer $(n,N,L\lambda,\lambda)$-DI code that achieves a linearithmic rate 
\(
\dot{R}=\frac{1}{2}.
\)
\end{theorem}
\begin{proof}
We follow the steps in the proof of Theorem~\ref{thm:double-layer}. Start by calculating the probability of correct identification of a message $i\in[N]$ encoded into word $\vec{o}_{s^{L}}$. At the layer $\ell$ the corresponding decoder fulfils:
\begin{align}
\!\!\Pr&(\cD_i^{(\ell)}|\vec{u}_i)=\Pr\left(\|\Pi_{s_{\ell-1}\to s_{\ell}}\vec{Y}-\vec{o}_{s^\ell}\|\leq\sigma\log n|\vec{o}_{s^{L}}\right)\nonumber \\
&=\Pr\left(\|\Pi_{s_{\ell-1}\to s_{\ell}}\vec{o}_{s^{L}}-\vec{o}_{s^\ell}+\sigma \Pi_{s_{\ell-1}\to s_{\ell}}\vec{Z}\|\leq \sigma\log n\right)\nonumber\\
&=\Pr\left(\|\Pi_{s_{\ell-1}\to s_{\ell}}\vec{Z}\|\leq\log n\right)\label{eq:layer_error1}\\
&=1-2\Phi(-\log n)=1-\lambda,\nonumber
\end{align}
as $\Pi_{s_{\ell-1}\to s_{\ell}}\vec{o}_{s^{L}}=\vec{o}_{s^\ell}$. This correct identification probability is the same for all $L$ layers so, by the the union bound, the total missed identification probability is bounded for all $i\in[N]$ by:
\begin{equation}\label{eq:main_Pe1}
\begin{split}
\!\!P_{e,1}(\vec{u}_i):=\Pr\left(\bigcup_{\ell=1}^{L}\cD_i^{(\ell)^c}\vert\vec{u}_i\right)&\leq\sum_{\ell=1}^{L}\Pr\left(\cD_i^{(\ell)^c}\vert\vec{u}_i\right)\\
&=2L\Phi(-\log n)\\
&=L\lambda=:\lambda_1,
\end{split}
\end{equation}
where the superscript $c$ indicates the complementary set. Notice that, for any arbitrarily large $L$ independent on $n$, we have by Proposition~\ref{prop:projection} that $L\lambda=\exp [-O(\log n)^2]$ which vanishes for large values of $n$, so $\lambda_1$ is correctly bounded.

To bound the probability of false identification, we choose $d=3\sigma\log n$ as the minimum projection distance defining the $d$-angle-dense arrangements at each layer. Then, if a word $u_j=\vec{o}_{\tilde{s}^{L}}$ is sent and $\vec{u}_i=\vec{o}_{s^{L}}$ is tested, the probability of error a layer $\ell$ where $s_\ell\neq\tilde{s}_\ell$ is given by
\begin{equation}\label{eq:layer_error2}
\begin{split}
\Pr&(\cD_i^{(\ell)}|\vec{u}_j)=\Pr\left(\|\Pi_{s_{\ell-1}\to s_{\ell}}\vec{Y}-\vec{o}_{s^\ell}\|\leq\sigma\log n|\vec{o}_{\tilde{s}^{L}}\right)\\
&\!=\Pr\left(\|\Pi_{s_{\ell-1}\to s_{\ell}}\vec{o}_{\tilde{s}^{L}}-\vec{o}_{s^\ell}+\sigma \Pi_{s_{\ell-1}\to s_{\ell}}\vec{Z}\|\leq \sigma\log n\right)\\
&\!\leq\Pr\left(\|\Pi_{s_{\ell-1}\to s_{\ell}}\vec{o}_{\tilde{s}^{L}}-\vec{o}_{s^\ell}\|-\sigma\|\Pi_{s_{\ell-1}\to s_{\ell}}\vec{Z}\|\leq\sigma\log n\right)\\
&\!\stackrel{(a)}{\leq}\Pr\left(\|\Pi_{s_{\ell-1}\to s_{\ell}}\vec{Z}\|\geq\log n\right)\\
&=2\Phi(-\log n)=\lambda,
\end{split}
\end{equation}
where we have used in $(a)$ that $\|\Pi_{s_{\ell-1}\to s_{\ell}}\vec{o}_{\tilde{s}^{L}}-\vec{o}_{s^\ell}\|\geq d-\delta\geq 2\sigma\log n$ [using Lemma~\ref{lemma:proj_dist} with the same argument as that in Equation~\eqref{eq:double-layer-delta-step}]. Now, we can simply upper bound the false identification error probabilities in all the other layers by 1. Then, the total false identification error probability is bounded for all $i\neq j\in[N]$ by:
\begin{equation}\label{eq:main_Pe2}
P_{e,2}(\vec{u}_i|\vec{u}_j)=\Pr\left(\cD_i|\vec{u}_j\right)\leq\Pr\left(\cD_i^{(\ell)}|\vec{u}_j\right)\leq\lambda:=\lambda_2.
\end{equation}

Both errors are correctly bounded. It is only left to count the number of elements in the code. From Proposition~\ref{prop:size_bound} [or, equivalently, using Equation~\eqref{eq:encoding_code_size}], and remembering the choice of radius $r_\ell=(cn)^{1/2^\ell}$ at each layer $\ell$, we have
\begin{equation}
\begin{split}
\log N_\ell &\geq\frac{n+1-\ell}{2}\left[\frac{\log n}{2^\ell}-\log\log n+O(1)\right]\\
&\geq\frac{n\log n}{2}\left[\frac{1}{2^\ell}-O\left(\frac{\log \log n}{\log n}\right)\right].
\end{split}
\end{equation}

The total number of elements in our code is given by $N=\prod_{\ell=1}^{L}N_{\ell}$ which implies $\log N=\sum_{\ell=1}^{L}\log N_{\ell}$. Then, the linearithmic rate is bounded by
\begin{equation}\label{eq:proof_rate}
\begin{split}
\dot R (n):&=\frac{\log N}{n\log n}=\frac{1}{n\log n}\sum_{\ell=1}^{L}\log N_\ell\\
&\geq\frac12\sum_{\ell=1}^{L}\frac{1}{2^\ell}-O\left(\frac{\log\log n}{\log n}\right).
\end{split}
\end{equation}

Notice that by choosing a large enough $L$ one can get the sum $\sum_{\ell=1}^{L}\frac{1}{2^\ell}$ to be arbitrarily close to $1$. Thus, as the second lower order term vanishes as $n\to\infty$, given any $\epsilon>0$ we can find a sufficiently large value of $L$ and a threshold $n_0$ such that for all $n>n_0$ we have $\dot R(n)\geq\frac12-\epsilon$. In other words, $\dot R=\frac12$ is an achievable rate.
\end{proof}

\begin{corollary}\label{cor:capacity}
The deterministic identification capacity of an additive white Gaussian noise channel is $\dot{C}_\text{DI}(\cG)=\frac12$.
\end{corollary}
\begin{proof}
$\dot R=\frac12$ is achievable, so $\dot{C}_\text{DI}(\cG)\geq\frac12$. We know from \cite{DI-fading,RRGauss-arXiv} that $\dot{C}_\text{DI}(\cG)\leq\frac12$. So, necessarily, $\dot{C}_\text{DI}(\cG)=\frac12$.
\end{proof}

\subsection{Capacity-achieving universal code}\label{ssec:universal}
In this section we show that it is possible to construct a capacity-achieving DI code whose encoder and decoder do not depend on the channel parameters. Codes with this property are commonly referred to as \emph{universal codes}. In other words, the same coding and decoding rules achieve the capacity of DI over any AWGN channel, without requiring knowledge of the signal-to-noise ratio and, in our case, even without knowledge of the power constraint.

Universality has been extensively studied for Shannon's message transmission. For discrete memoryless and finite-state channels, universal coding schemes that operate without knowledge of the channel transition probabilities are well established \cite{CK:book2011,Ziv-universal}. For Gaussian channels, there exist universal constructions which remove the dependence of the decoder on the noise variance but cannot achieve capacity \cite{LN:General-universal,Gaussian-universal} unless feedback is used to adapt the transmission rate to the realized channel conditions \cite{rateless-codes}.

Our universal construction achieves capacity for deterministic identification over Gaussian channels, and even goes one step further: not only does the decoder operate without knowledge of the noise variance, but the encoder itself does not require knowledge of the power constraint. Nevertheless, the resulting code words remain inside the admissible input region, for sufficiently large block lengths. 

To establish this result, we first identify the steps in the previous proof where knowledge of the channel parameters (noise variance and power constraint) is needed. The first such step appears in the choice of radii. Recall that the radius of the angle-dense arrangement at layer $\ell$ was defined as $r_\ell=(cn)^{1/2^\ell}$ with the crucial requirement $c<P$. If this condition is not satisfied, the resulting code words may fall outside the input space. Hence the choice of radii in the previous construction depends on the value of $P$. The second dependence arises in the choice of the projective decoding range, which depends on the noise level. In order to accept a word at each layer, its projection onto the corresponding path vector must lie within a distance $\sigma \log n$ of a given centre [see Equation~\eqref{eq:decoder}]. Other parameters in the construction, such as $d=3\sigma\log n$, $\theta_\ell$, or $\delta$, also involve $\sigma$ and/or $P$. However, these dependencies ultimately stem from the definitions of the radii and the decoding range described above. 

\noindent \textbf{Encoding:}
Define for a small value $b>0$ the following radius at each layer $r_\ell=n^{(1-b)2^{-\ell}}$. This definition avoids the constant parameter $c$ (dependent on $P$), at the cost of slightly reducing the scaling of the radii at each layer, which will not ultimately affect the achievable rate. Then, mimic the angle-dense strategy from the previous section. For sufficiently large $n$ the code words $\vec o_{s^L}$, with $s^L=(s_1,\dots,s_L)$ and $s_\ell\in [N_\ell]$ for $\ell\in[L]$, are in the input space restricted by the power constraint:
\begin{equation}\label{eq:uni_radi}
\sum_{\ell=1}^{L}(n^{1-b})^{2/2^\ell}=O(n^{1-b})<nP.
\end{equation}

\noindent \textbf{Decoding:}
Simply redefine the decoding range at each layer $\ell$ as $\log n$ (removing $\sigma$). The decoder corresponding to the message $i\in[N]$ encoded into the code word $\vec o_{s^L}$ does not depend on $\sigma$ at any layer $\ell$:
\begin{equation}\label{eq:uni_decoder}
\cD_{i}^{(\ell)}=\{\vec{y}\in\RR^n:\|\Pi_{s_{\ell-1}\to s_{\ell}}\vec{y}-\vec{o}_{s^\ell}\|\leq\log n\}.
\end{equation}
Accordingly, the choice of the projective distance also needs to be changed to $d=3\log n$. 
\medskip

Although these modifications can have a negative impact on the value of the rate and the error performance, we will show they do not affect at first order, so the rate is asymptotically unchanged and so are the errors. The result is a good capacity-achieving universal DI code, with errors vanishing for $n\to\infty$.

\begin{theorem}\label{thm:uni_code}
Let $\lambda'=2\Phi(-\frac{\log n}{\sigma})$. Then, with the coding construction above, a good universal $(n, N,L\lambda',\lambda')$ DI code exists, achieving a linearithmic rate $\dot R\geq\frac12$.
\end{theorem}

\begin{proof}
We follow the proof of Theorem~\ref{thm:main}. The probability of correct identification of the word $\vec{u}_i=\vec{o}_{s^{L}}$ at $\ell$ is
\begin{equation}\label{eq:uni_layer_error1}
\begin{split}
\Pr&(\cD_i^{(\ell)}|\vec{u}_i)=\Pr\left(\|\Pi_{s_{\ell-1}\to s_{\ell}}\vec{Y}-\vec{o}_{s^\ell}\|\leq\log n|\vec{u}_i\right)\\
&=\Pr\left(\|\Pi_{s_{\ell-1}\to s_{\ell}}\vec{u}_i-\vec{o}_{s^\ell}+\sigma \Pi_{s_{\ell-1}\to s_{\ell}}\vec{Z}\|\leq \log n\right)\\
&=\Pr\left(\|\Pi_{s_{\ell-1}\to s_{\ell}}\vec{Z}\|\leq\frac{\log n}{\sigma}\right)\\
&=1-2\Phi\left(-\frac{\log n}{\sigma}\right)=1-\lambda',
\end{split}
\end{equation}
and as this error probability is the same for all $L$ layers we get [similarly as in Equation~\eqref{eq:main_Pe1}] 
\begin{equation}\label{eq:uni_Pe1}
P_{e,1}\leq L\lambda'=O(e^{-(\log n)^2}).
\end{equation}
The error scaling is not affected, as $\sigma$ (while unknown) is a constant, meaning that the order of decay of the error with $n$ actually does not change,
thus the error still vanishes as $n\to\infty$ for any value of $\sigma$. 

For the probability of false identification we have chosen $d=3\log n$ as the minimum projection distance that defines our $d$-angle-dense arrangements at each layer. Then, if a word $\vec u_j=\vec{o}_{\tilde{s}^{L}}$ is sent and $\vec{u}_i=\vec{o}_{s^{L}}$ is tested, the probability of committing an error at a layer $\ell$ where $s_\ell\neq\tilde{s}_\ell$ is given by
\begin{equation}\label{eq:uni_layer_error2}
\begin{split}
\!\!\!\!\Pr&(\cD_i^{(\ell)}|\vec{u}_j)=\Pr\left(\|\Pi_{s_{\ell-1}\to s_{\ell}}\vec{Y}-\vec{o}_{s^\ell}\|\leq\log n|\vec{o}_{\tilde{s}^{L}}\right)\\
%&=\Pr\left(\|\Pi_{s_{\ell-1}\to s_{\ell}}\vec{u}_j-\vec{o}_{s^\ell}+\sigma \Pi_{s_{\ell-1}\to s_{\ell}}\vec{Z}\|\leq \log n)\\
&\leq\Pr\left(\|\Pi_{s_{\ell-1}\to s_{\ell}}\vec{o}_{\tilde{s}^{L}}-\vec{o}_{s^\ell}\|\!-\!\sigma\|\Pi_{s_{\ell-1}\to s_{\ell}}\vec{Z}\|\!\leq\!\log n\right)\\
&\leq\Pr\left(\|\Pi_{s_{\ell-1}\to s_{\ell}}\vec{Z}\|\geq\frac{\log n}{\sigma}\right)\\
&=2\Phi\left(-\frac{\log n}{\sigma}\right)=\lambda'.
\end{split}
\end{equation}
Then, the total false identification error probability is given for all $i\neq j\in[N]$ by 
\(
P_{e,2}(\vec{u}_i,\vec{u}_j)\leq\Pr\left(\cD_i^{(\ell)}|\vec{u}_j\right)\leq\lambda'.
\)

Both errors are correctly bounded. It is only left to count the number of elements in the code with Proposition~\ref{prop:size_bound}. At layer $\ell$ we have an arrangement of radius $r_\ell=n^{(1-b)/2^{\ell}}$ on a ball of dimension $n-\ell$, therefore:
\begin{equation}
\begin{split}
\log N_\ell &\geq\frac{n-\ell}{2}\left[\frac{1-b}{2^\ell}\log n-\log\log n+O(1)\right]\\
&\geq\frac{n\log n}{2}\left[\frac{1-b}{2^\ell}-O\left(\frac{\log \log n}{\log n}\right)\right].
\end{split}
\end{equation}
Therefore, the linearithmic rate of the code is given by 
\begin{equation}\label{eq:uni_rate}
\dot R(n)=\sum_{\ell=1}^{L}\frac{\log N_\ell}{n\log n}
\geq\frac{1-b}{2}\sum_{\ell=1}^{L}\frac{1}{2^\ell}-O\left(\frac{\log\log n}{\log n}\right).
\end{equation}

Choosing $L$ large enough one can get the sum $\sum_{\ell=1}^{L}\frac{1}{2^\ell}$ to be arbitrarily close to $1$, and $b>0$ can also be chosen as small as one wishes. Therefore, given any $\epsilon>0$ one can find a sufficiently large value of $L$, a threshold $n_0$, and choose small enough $b$ such that $\dot R(n)\geq\frac12-\epsilon$ for all $n>n_0$, meaning that $\dot R=\frac12$ is an achievable rate. Thus, this universal code achieves the channel capacity.
\end{proof}

%%%%%%%%%%%%%%%%%%%%%%%%
%%%%%%%%%%%%%%%%%%%%%%%%
%%%%%%%%%%%%%%%%%%%%%%%%
%%%%%%%%%%%%%%%%%%%%%%%%

\section{Rate-reliability tradeoff}\label{sec:RR_tradeoff}
In the previous sections we studied a specific construction that achieves capacity, ensuring that the first-order rate asymptotically attains the optimal value while maintaining a vanishing probability of error. In this section we study the rate–reliability tradeoff extending these methods to general error regimes. In other words, we investigate the achievable rates under different fixed error requirements rather than only requiring the probability of error to vanish.

For DI, this tradeoff was first studied for channels with continuous input and discrete output in \cite{CDBW:Reliability-TCOM,DI-steins}, and later extended to general linear Gaussian channels in \cite{RRGauss-arXiv}. The direct results in these works rely on typicality-based arguments and exhibit the expected $\frac{1}{4}$ gap between the lower and upper bounds on the rate–reliability functions.
We show that the construction introduced in the previous sections improves the lower bound in \cite{RRGauss-arXiv}, matching at first order the converse bound derived there in the regime of small error exponents. That upper bound derived in \cite{RRGauss-arXiv} is:
\begin{equation}\label{eq:RR_upper_bound}
R(n)\leq\frac12\log\frac{1}{E(n)}+\frac12\log\frac{8P}{\sigma^2},
\end{equation}
where $R(n):=\frac1n \log N$ is the linear rate, and $E(n)=\min\{E_1(n),E_2(n)\}$ is the minimum of the error exponents, defined as $E_1(n):=-\frac1n\log\lambda_1$ and $E_2(n):=-\frac1n\log\lambda_2$, and fulfils $E(n)\geq\frac{\ln16}{n}$.

Notice that, for constant error exponents, the rate can only be linear. Indeed, the linear rate in Equation~\eqref{eq:RR_upper_bound} is upper bounded by a constant. However, when the error exponent is a decreasing function in $n$ (that means, the errors vanish slower than exponentially), the linear rate diverges in $n$, meaning that we have a superlinear scaling of the rate. Indeed, in the limiting case of the slowest possible vanishing speed of the errors, when $E(n)=\Omega(1/n)$, we recover the linearithmic behaviour with the prefactor $\frac12$ in the first order term.

\begin{theorem}\label{thm:RR}
Given an AWGN channel with input power constraint $P>0$ and noise variance $\sigma^2>0$, one can construct an $(n,N,L2^{-nE(n)},2^{-nE(n)})$ DI code for all $E(n)\leq E_0:=\frac{9P}{\sigma^2}$ with $L\gg1$ angle-dense layers achieving a linear rate of
\[
R\geq \frac12\log\frac{1}{E(n)}+\frac12\log\frac{P}{9L\sigma^2}.
\]
\end{theorem}
\begin{proof}
Start by observing that, in our construction, the error performance is dictated by the choice of the accepted range in the decoders at each layer, let us call it $\sigma x(n)$ [see Equation~\eqref{eq:decoder}, there $x(n)=\log n$]. Then, the first type of error at each layer is bounded by $\lambda=2\Phi[-x(n)]$, cf.~Equation~\eqref{eq:layer_error1}. Recalling Proposition~\ref{prop:projection}, and by choosing $x(n)=\sqrt{2nE(n)}$ the following exponential error expression is obtained:
\begin{equation}
\lambda=2\Phi\left(-\sqrt{2nE(n)}\right)\leq 2^{-nE(n)}.
\end{equation}
The total probability of missed identification in a code of $L$ layers is therefore bounded by $\lambda_1\leq L\lambda$. Then, for any arbitrarily large constant $L$, the error converges to 0 when $n\to\infty$.

The second type of error (false identification) can be similarly bounded. Choosing $d=3\sigma x(n)$ one observes [cf.~Equation~\ref{eq:main_Pe2}], that $\lambda_2:=\lambda=2^{-nE(n)}$. Thus, the choices of $x(n)=\sqrt{2nE(n)}$ and $d=3\sigma x(n)$ produce a code with errors that can be written in the desired exponential form. The decoder at each layer is:
\begin{equation}\label{eq:RR_decoder}
\cD_{i}^{(\ell)}=\left\{\vec{y}\in\RR^n:\|\Pi_{s_{\ell-1}\to s_{\ell}}\vec{y}-\vec{o}_{s^\ell}\|\leq\sigma\sqrt{2nE(n)}\right\}.
\end{equation}

Let us study now how the previous definitions fix the code structure. From Lemma~\ref{lemma:angle} one can extract 
\begin{equation}\label{eq:sin1}
    \sin \theta_\ell\simeq\sqrt{\frac{2d}{r_\ell}},
\end{equation}
where we have used the small-angle approximation $\sin\frac{\theta_{\ell}}{2}\approx\frac12\sin\theta_{\ell}$, which becomes asymptotically exact as $n\to\infty$ since $d=o(r_\ell)$, i.e.~$r_\ell$ grows strictly faster than $d$ with the block length. 

Similarly, to make the second type of error work with the previous choice of $d$, it is necessary that $\delta\leq \sigma x(n)$, such that $\|\Pi_{s_{\ell-1}\to s_{\ell}}\vec{o}_{\tilde{s}^{L}}-\vec{o}_{s^\ell}\|\geq d-\delta\geq 2\sigma\log n$. This is ensured if 
\begin{equation}\label{eq:sin2}
\sqrt{L}r_{\ell+1}\sin\theta_\ell=\sigma x(n)
\end{equation}
Then, equating the expressions \eqref{eq:sin1} and \eqref{eq:sin2}, one obtains the following relation between the radius at one layer and the radius at the next:
\begin{equation}\label{eq:RR_hierarchy}
    r_{\ell+1}=\sqrt{\frac{r_\ell\sigma x(n)}{6L}}.
\end{equation}
If one now fixes $r_1=\sqrt{\frac{Pn}{2}}$ as the radius at the first layer, the expression above defines the radius at all other layers. The following relation is obtained:
\begin{equation}\label{eq:RR_radius}
r_\ell = \sqrt{\rho n E(n)}\left[\frac{P}{2\rho E(n)}\right]^\frac{1}{2^\ell}\!\!,\,\,\,\,\text{with}\,\,\,\rho:=\left(\frac{\sigma}{3\sqrt{2}L}\right)^2\!\!.
\end{equation}

As the words have to be inside the input space, it is required that $\sum_{\ell=1}^L r_\ell^2\leq nP$, which is necessary fulfilled if $L r_2^2\leq nP-r_1^2=\frac{nP}{2}$. Then, as
\begin{equation}
L r_2^2=L\rho n E(n)\sqrt{\frac{P}{2\rho E(n)}}=n\frac{\sigma}{6}\sqrt{PE(n)},
\end{equation}
it is necessary that
\begin{equation}
\frac{\sigma}{6}\sqrt{PE(n)}\leq\frac{P}{2}\implies E(n)\leq E_0:=\frac{9P}{\sigma^2}.
\end{equation}

One can use the considerations above to create a good DI code. Let us now analyse its rate performance as a function of the error exponent $E(n)$. Similarly as in the previous section, the rate can be bounded using Equation~\eqref{eq:encoding_code_size} (extracted from Proposition~\ref{prop:size_bound}) and the radius at each layer, given in Equation~\eqref{eq:RR_radius}. Then,
\begin{equation}
\log N_\ell\geq\frac{n-(\ell-1)}{2}\log\left[\frac{2\sqrt{\rho n E(n)}}{3\sigma\sqrt{2nE(n)}}\left(\frac{P}{2\rho E(n)}\right)^\frac{1}{2^\ell}\right].
\end{equation}
Separate now the right-hand side of the equation above into a part with the linear prefactor and a part without it $\frac{n}{2}\log(\dots)-\frac{\ell-1}{2}\log(\dots)$. Then, the negative (second) term can be bounded by the error exponent lower bound $E(n)\geq\frac1n$ and the expression becomes:
\begin{equation}
\begin{split}
\log N_\ell&\geq\!\frac{n}{2}\left[\frac{1}{2^\ell}\log\left(\frac{9PL^2}{\sigma^2 E(n)}\right)\!-\!\log\left(9L\right)\right]\!-\!O\left[(\log n)^\frac{1}{2^\ell}\right]\!.
\end{split}
\end{equation}
The linear rate will then be bounded by:
\begin{align}
R(n):=\sum_{\ell=1}^L\frac{\log N_\ell}{n}\geq\, &\frac12\left(\sum_{\ell=1}^L\frac{1}{2^\ell}\right)\!\!\left(\log\frac{1}{E(n)}+\log\frac{P}{\sigma^2}\right)\nonumber\\
&-\frac{\log 9L}{2}-O\left(\frac{\sqrt{\log n}}{n}\right).\label{eq:RR_rate}
\end{align}

Now, for any $\epsilon>0$ there exists a threshold value $n_0$ and large enough $L$ such that for all $n>n_0$ one finds 
\begin{equation}
R(n)\geq\frac12 \left(\log\frac{1}{E(n)}+\log\frac{P}{9L\sigma^2} \right)-\epsilon.
\end{equation}
Therefore, the claimed rate is indeed achievable.
\end{proof}

Finally, notice that in the regime of small error exponents the term $\frac{1}{2}\log\frac{1}{E(n)}$ in the rate above will dominate. In particular, if $E(n)$ is a decreasing function of $n$, the second term will be of lower order. Thus, in this regime, the provided lower bound to the rate-reliability function matches the known upper bound [given here in Equation~\eqref{eq:RR_upper_bound}] at first order.

\section{Conclusions}\label{sec:conclusions}
In this paper, we have closed the gap between the lower and upper bounds on the linearithmic capacity of deterministic identification over additive white Gaussian noise channels. This gap has been repeatedly observed in DI over channels with continuous input \cite{CDBW:DI_classical, DI-fading}, a broad and important class of channels relevant both for practical applications \cite{DI-poisson_mc, 6G_Book, 6G_Book_2} and also because they exhibit linearithmic rate scaling, in contrast to the linear scaling of DI over DMCs \cite{SPBD:DI_power} and the linear scaling of message transmission \cite{Shannon:TheoryCommunication}.

In only a few special cases, where particular symmetries of the channel enable the design of optimal decoding strategies beyond typicality-based arguments, has the capacity been exactly characterized \cite{CDBW:DI_classical, qhtl_ISIT}. In more general settings, existing constructions rely on typical sets and fall short of achieving the known upper bounds. Indeed, it was conjectured in earlier work that the upper bound was probably tight and that the use of typicality in the code constructions may be the fundamental limitation preventing the attainment of capacity.

The construction proposed in \cite{galaxy-codes} introduced a novel geometric approach based on projection properties of the noise. Combining this with typicality-based codes led to an improved lower bound on the linearithmic capacity for Gaussian channels, increasing it from $\frac{1}{4}$ to $\frac{3}{8}$. Building on this idea, we have developed a fully projection-based construction that entirely eliminates the need for typicality and achieves the upper bound, establishing that $\dot C_{\text{DI}}(\cG) = \frac{1}{2}$.

A notable consequence of this result is that the DI capacity does not depend on the Gaussian channel parameters. While this is consistent with previous results, for instance, it was shown in \cite{CDBW:DI_classical} that the DI capacity of channels with continuous input and discrete output depends only on the geometric structure of the output probability set; it stands in sharp contrast with classical Shannon transmission. In the latter case, capacity depends explicitly on information-theoretic quantities such as mutual information and divergences. Specifically, for AWGN channels, it depends on the signal-to-noise ratio. This dependence makes the design of universal capacity-achieving codes infeasible in general transmission settings. In contrast, for DI over Gaussian channels, where channel parameters only affect second-order terms, we have shown that universal capacity-achieving codes do exist.

Finally, the gap between lower and upper bounds also appears in the study of the rate-reliability tradeoff \cite{CDBW:Reliability-TCOM, RRGauss-arXiv}. By analysing the tradeoff between rate and reliability in the proposed construction over a general range of error regimes, we have shown that this gap can also be closed in the regime of small error exponents (slowly vanishing errors), where the achievable performance matches the upper bound to first order. Importantly, this particular regime (small error exponents) is the regime of primary interest, as it is required to achieve superlinear scalings of the message size. Indeed, we observed that if the error exponent is a constant, the rate reduces to linear scaling.

A natural next step is to extend the present construction beyond AWGN channels. In particular, it would be interesting to investigate channels with continuous inputs and discrete outputs. One candidate is a Bernoulli-type channel model, which may serve as a tractable representative for broader families including the Poisson channel, commonly used to model emerging communication systems.

Another important direction is the computational complexity of the proposed scheme. Since decoding reduces to a collection of projections and distance computations, the construction appears amenable to polynomial-time implementation. Finally, the hierarchical structure of the code, which allows positive rate contributions from small input regions, may also have implications for secrecy and covertness.

%\clearpage

\bibliographystyle{ieeetr}
\bibliography{ID}

\end{document}